\documentclass[notitlepage,prx,longbibliography,letter,reprint,footinbib]{revtex4-1}

\usepackage[draft]{todonotes}
\usepackage{amsthm}
\usepackage{amsfonts}
\usepackage{amsmath}
\usepackage{enumerate}
\usepackage{mathrsfs}
\newcommand{\rood}[1]{\ignorespaces}
\newcommand{\comment}[1]{} 

\newcommand{\be}{\begin{eqnarray*}}
\newcommand{\ee}{\end{eqnarray*}}
\newcommand{\beq}{\begin{eqnarray}}
\newcommand{\eeq}{\end{eqnarray}}

\newcommand{\ket}[1]{\left|{#1}\right\rangle}

\newcommand{\bk}{\mathbf{k}}

\newcommand{\dd}{\mathrm{d}}

\newcommand{\mcs}{\mathcal{S}}
\newcommand{\cat}{\mathscr{C}}

\newtheorem{lemma}{Lemma}[section]
\newtheorem{theorem}{Theorem}[section]
\newtheorem{claim}{Claim}[section]
\begin{document}
\title{Periodic Table for Floquet Topological Insulators}
\author{Rahul Roy}
\author{Fenner Harper}
\affiliation{Department of Physics and Astronomy, University of California, Los Angeles, California USA}
\date{\today}
\begin{abstract}
Dynamical phases with novel topological properties are known to arise in driven systems of free fermions. In this paper, we obtain a `periodic table' to describe the phases of such time-dependent systems, generalizing the periodic table for static topological insulators. Using K-theory, we systematically classify Floquet topological insulators from the ten Altland-Zirnbauer symmetry classes across all dimensions. We find that the static classification scheme described by a group $G$ becomes $G\times G$ in the time-dependent case, and interpret the two factors as arising from the bipartite decomposition of the unitary time-evolution operator. Topologically protected edge modes may arise at the boundary between two Floquet systems, and we provide a mapping between the number of such edge modes and the topological invariant of the bulk.
\end{abstract}
\maketitle

\section{Introduction}

 The discovery of topological insulators and the theoretical and experimental activity that it inspired has led to major advances in our understanding of zero-temperature gapped phases~\cite{Hasan:2010ku,Qi:2011hb}.  While the first new systems to be discovered were specific topological phases of insulators and superconductors in one to three dimensions~\cite{Kane:2005hl,Moore:2007gq,Roy:2009kk,Fu:2007io,Roy:2008ty,Roy:2006uy,Schnyder:2008ez}, these were eventually arranged into a `periodic table', which extended the classification to all dimensions and symmetry classes~\cite{Kitaev:2009vc}. This unifying approach revealed a remarkable underlying periodicity, using connections between K-theory and Bott periodicity on the one hand, and free fermionic topological phases with symmetries on the other. 
 
The generalized topological insulators that this classification scheme describes exhibit robust, topologically protected edge modes in the presence of a boundary, and are characterized by invariant integers encoded in the topology of their wavefunctions. In this way, the periodic table captures the complete set of bulk-edge connections between bulk Hamiltonians and their protected edge states. Equivalently, one can interpret the periodic table as expressing the connection between the unitary time evolution of a constant Hamiltonian (evaluated after some time $T$), and the corresponding edge eigenstates. In this picture, the periodic table may be regarded as part of a more general framework of topological bulk-edge connections between unitary time evolution operators and protected edge modes.  When the Hamiltonians involved are no longer constrained to be time-independent, new types of bulk-edge connection may occur. In this paper, we seek to capture the structure of these dynamical bulk-edge connections by constructing a generalized periodic table for free fermionic systems with time-dependent Hamiltonians. 

Among our motivations for studying these systems is the set of (time-periodic) Floquet topological insulators that have recently been the subject of much experimental and theoretical effort [see Refs.~\onlinecite{Cayssol:2013gk,Bukov:2015gu} for a review]. Some of these efforts have considered using periodic driving to force a system into a topological state \cite{Yao:2007iy,Oka:2009kc,Inoue:2010iz,Lindner:2011ip,Kitagawa:2011fj,Lindner:2013gh,Grushin:2014gt,PerezPiskunow:2014iy,FoaTorres:2014gx,Calvo:2015bn,Iadecola:2015dg}, and significant experimental progress has be made in this direction in photonic systems \cite{Rechtsman:2013fe,Kitagawa:2012gl} and using ultracold atoms \cite{Jotzu:2015kz,JimenezGarcia:2015kd}. Floquet states (albeit non-topological ones) have also been observed in the solid state on the surfaces of topological insulators \cite{Wang:2013fe,Fregoso:2013di}. Other recent work has demonstrated the possibility of generating intrinsically dynamical topological phases that cannot be realized in static systems \cite{Kitagawa:2010bu,Jiang:2011cw,Rudner:2013bg,Thakurathi:2013dt,Asboth:2014bg,Titum:2015wj,Nathan:2015bi,Carpentier:2015dn,Titum:2015fl,Fruchart:2016hk}. 

Although we will make connections to Floquet theory, our approach provides a description of time-dependent topological phases in a manner that does not require time periodicity. Instead, we consider equivalence classes of unitary time-evolution operators in general, and focus on the instantaneous topological edge states that might exist in a system after a particular time evolution. Our main result will be the production of a generalized periodic table of Floquet topological insulators, which may be found in Table.~\ref{tab:AZ_class_2g}. In the process, we find many new Floquet topological phases that have not been considered before, and provide a general and unifying description for all symmetry classes and dimensions. As in the case of (static) topological insulators, this picture provides a connection between the manifestations of Bott periodicity in K-theory and the topological phases of driven free fermionic systems, describing both the strong and weak invariants of the system.

Some elements of our generalized periodic table have appeared in the context of Floquet systems elsewhere in the literature. Notably, previous work has considered dynamical topological phases in 1D chains with emergent Majorana fermions \cite{Jiang:2011cw,Thakurathi:2013dt}, 2D systems without symmetries \cite{Rudner:2013bg} and driven analogues of the 2D time-reversal invariant topological insulators \cite{Carpentier:2015dn}. Topological phases of 1D chiral lattices have also been described in Ref.~\onlinecite{Asboth:2014bg}, albeit using a different definition of chiral symmetry than will be used in this work. In addition, Ref.~\onlinecite{Nathan:2015bi} describes a band singularity approach to the characterisation of Floquet topological phases, introducing new results for 3D systems with time-reversal symmetry. After the completion of our work we discovered Ref.~\onlinecite{Fruchart:2016hk}, which extends the formulation of strong topological invariants for classes A and AIII to all dimensions. While our work does not rely on invariants for classification and discusses several other cases, our results seem to be consistent with these existing discussions and incorporates them into a general, unifying periodic table. Our results also capture the complete set of strong and weak topological invariants in each case.

In this paper we consider noninteracting systems, but the ideas we outline also develop some of the intuition required for the study of interacting topological phases, a topic that has been the focus of much recent study, both by the current authors and others \cite{Else:2016tj,Potter:2016tb,vonKeyserlingk:2016vq,Roy:2016wd}. Importantly, the statements we make in the noninteracting case can, to a great extent, be made mathematically precise, while arguments for interacting systems necessarily require a certain amount of conjecture. In this way, we hope that this paper will provide a useful corroboration of the ideas introduced in Ref.~\onlinecite{Roy:2016wd}.
 
The outline of this paper is as follows. In Sec.~\ref{sec:preliminary}, we introduce unitary evolution operators in the context of time-dependent systems, and establish the homotopy formalism that we will require throughout the text. In Sec.~\ref{sec:unitary_decomposition} we introduce unitary loops and explain how a general unitary evolution may be deformed into a unitary loop followed by a constant Hamiltonian evolution, a theorem that is central to our approach. We go on to classify unitary loops for the Altland-Zirnbauer (AZ) symmetry classes in Sec.~\ref{sec:loop_classification}, before relating this classification scheme to general unitaries and edge modes in Sec.~\ref{sec:discussion}. Finally, we give some concluding remarks in Sec.~\ref{sec:conclusions}. In order to aid ease of reading, we have omitted some of the more mathematical sections from the main text. These may be found in the appendices.


\section{Unitary Time Evolution Operators and their Properties\label{sec:preliminary}}
\subsection{Time-dependent Quantum Systems}
The aim of this paper is to classify the novel types of topological edge mode that can arise in a quantum system after it has evolved in time due to some time-dependent Hamiltonian $H(t)$. In general, instantaneous eigenstates satisfy the time-dependent Schr\"{o}dinger equation and evolve in time through the unitary transformation
\begin{equation}
\ket{\psi(t)}=\mathcal{T}\exp\left[-i\int_0^tH(t')\dd t'\right]\ket{\psi(0)}\equiv U(t)\ket{\psi(0)},\label{eq:unitary_definition}
\end{equation}
with $\mathcal{T}$ the time ordering operator. $U(t)$ is the time evolution operator, and, being unitary, has eigenvalues that lie on the unit circle in the complex plane. We write these eigenvalues as $e^{-i\epsilon(t)t}$, and focus on the instantaneous quasienergies given by $\{\epsilon(t)\}$, taken to lie in the range $-\pi/t<\epsilon(t)\leq \pi/t$. In a spatially periodic system, the instantaneous quasienergies form bands labelled by the momentum $\bk$ and a band index. In some ways, these bands bear a resemblance to the ordinary bands of a (static) periodic Hamiltonian, although we will find that the periodic nature of the quasienergy spectrum generally allows for a much richer structure.

We are particularly interested in the quasienergy spectrum after evolution through some time period $T$, and we write the quasienergies at $t=T$ simply as $\epsilon$. At this point, a system with an open boundary should have a similar quasienergy spectrum to the corresponding periodic system, with the possible addition of energy levels in the gaps between the bulk bands. The existence of gap states indicates the presence of topologically protected edge modes, which we aim to classify in this text.

Most previous work in this area has focussed on Floquet systems: those whose time-dependent Hamiltonians satisfy $H(t+T)=H(t)$ for some time period $T$. In a Floquet system, we can use an analogy of Bloch's theorem to write the instantaneous eigenstates as $\ket{\psi(t)}=e^{-i\epsilon(t) t}\ket{\phi(t)}$ with $\ket{\phi(t+T)}=\ket{\phi(t)}$. In this way, after a complete time period, Floquet states simply pick up a phase, since $U(T)\ket{\psi(0)}=e^{-i\epsilon T}\ket{\psi(0)}$. It should be noted, however, that the time evolution operator $U(t)$ is generally \emph{not} periodic, even if it is derived from a periodic Hamiltonian.

Although Floquet theory provides a useful setting in which to discuss time-dependent systems, we emphasize that our conclusions will be much more general than this. We will make statements about the protected edge modes present in the quasienergy spectrum after the unitary evolution $U(T)$; whether or not the system is periodic beyond this point in time is unimportant.

\subsection{Particle-hole, Time-reversal and Chiral Symmetries}
In this paper, we are concerned with free fermionic systems that fall within the symmetry classes of the AZ classification scheme \cite{Heinzner:2005je,Altland:1997zz,Zirnbauer:1996cp}. These classes are distinguished by the presence or absence of two antiunitary symmetries and one unitary symmetry, as well as the general form of the relevant symmetry operators.

\begin{table}[t!]
\be
\renewcommand\arraystretch{1}\begin{array}{c|c}
\mbox{Symmetry Operator} & \mbox{Cartan Label $\mathcal{S}$}\\
\hline
P=\sigma_1\otimes\mathbb{I} & \mbox{BDI, D, DIII} \\
P=i\sigma_2\otimes\mathbb{I} & \mbox{CII, C, CI} \\
\hline
\theta=\mathbb{I} & \mbox{AI, BDI, CI}\\
\theta=\mathbb{I}\otimes i\sigma_2 & \mbox{AII, CII, DIII}\\
\hline
C=\sigma_3\otimes\mathbb{I} & \mbox{AIII}
\end{array}
\ee
\caption{Standard expressions for symmetry operators within each Altland-Zirnbauer (AZ) symmetry class. $\sigma_i$ are the Pauli matrices and $\mathbb{I}$ is the identity. \label{tab:symmetry_operators}}
\end{table}

In systems with particle-hole symmetry (PHS), there exists a PHS operator $\mathcal{P}=\mathcal{K}P$, where $\mathcal{K}$ is the complex conjugation operator and $P$ is unitary, that acts on the band Hamiltonian to give
\beq
PH(\bk,t)P^{-1}&=&-H^*(-\bk,t).
\eeq
Similarly, in systems with time-reversal symmetry (TRS), there exists a TRS operator $\Theta=\mathcal{K}\theta$, where $\mathcal{K}$ is again the complex conjugation operator and $\theta$ is unitary, that acts on the band Hamiltonian to give
\beq
\theta H(\bk,t)\theta^{-1}&=&H^*(-\bk,T-t).
\eeq
With this definition, we have assumed without loss of generality that $t=T/2$ is the point in time about which the Hamiltonian is symmetric. 

From the definition of the time evolution operator in Eq.~\eqref{eq:unitary_definition}, it follows that these symmetry operators, if present, act on $U(\bk,t)$ to give
\beq
PU(\bk,t)P^{-1}&=&U^*(-\bk,t)\label{eq:unitary_PHS}\\
\theta U(\bk,t)\theta^{-1}&=&U^*(-\bk,T-t)U^{\dagger *}(-\bk,T).\label{eq:unitary_TRS}
\eeq
The actions of each symmetry operator on the time evolution unitary are derived in Appendix~\ref{sec:symmetry_actions}.

If both TRS and PHS are present, there is an additional unitary symmetry $C=P\theta$ that acts on the Hamiltonian to give
\beq
CH(\bk,t)C^{-1}&=&-H(\bk,T-t).\label{eq:hamiltonian_CS}
\eeq
More generally, there may be a chiral symmetry (CS) operator $C\neq P\theta$ that acts on the Hamiltonian according to Eq.~\eqref{eq:hamiltonian_CS} even in the absence of PHS and TRS. This defines the tenth AZ symmetry class, labelled AIII. When acting on the time evolution unitary, the CS operator gives (derived in Appendix~\ref{sec:symmetry_actions})
\beq
CU(\bk,t)C^{-1}&=&U(\bk,T-t)U^\dagger(\bk,T).\label{eq:unitary_CS}
\eeq
We note that our definition of chiral symmetry for periodic systems is slightly different from that used in some previous works \cite{Asboth:2014bg,Nathan:2015bi}.

After a suitable basis transformation, $P$, $\theta$ and $C$ can always be written in certain standard forms, as shown in Table~\ref{tab:symmetry_operators}. We will implicitly assume these representations throughout this article. We write the set of unitaries that belong to each symmetry class as $\mathcal{U}^\mathcal{S}$, where $\mathcal{S}$ is the appropriate Cartan label. To simplify notation, we set $T=1$ from now on, so that $t\in[0,1]$. We will also often omit the explicit momentum and time dependence of a unitary operator $U(\bk,t)$ if the meaning is clear.

\subsection{Gapped Unitaries}
We are interested in the protected edge modes that may arise in the gaps of the quasienergy spectrum at the end of a unitary evolution if the system has a boundary. For this reason, we will restrict the discussion to consider only gapped unitaries, which we define to be unitary evolutions of the form in Eq.~\eqref{eq:unitary_definition}, which at their end point, $U(\bk,1)$, have at least one value of quasienergy in the closed system that no bands cross.\footnote{More general definitions may be used, but this definition suffices to capture the interesting aspects of the problem.} Importantly, we do not require that the instantaneous quasienergy spectrum be gapped for intermediate values of $t$ ($0<t<1$). We write the set of all such gapped unitaries within the AZ symmetry class $\mathcal{S}$ as $\mathcal{U}^\mathcal{S}_g$ and note that a unitary evolution of this form may be represented as a continuous matrix function $U(\bk,t)$ with $0\leq t \leq 1$ and $\bk$ taking values within the $d$-dimensional Brillouin zone, which we call $X$. It is clear that $U(\bk,t)$ evolves from the identity matrix at $t=0$.

The gap structure at the end of a unitary evolution will depend on the symmetry of the underlying Hamiltonian, and in general can be rather complicated. A schematic example of a gapped unitary evolution with PHS is shown in Fig.~\ref{fig:gapped_unitary}, which emphasizes both the nontrivial evolution and the quasienergy band structure at the end point. The most commonly considered quasienergy gaps are those at $\epsilon=0$ and $\epsilon=\pi$, since in many cases a generic gap can be moved homotopically [a term we define precisely below] to one of these points. We will discuss gap structures more generally in Sec.~\ref{sec:discussion}.

\begin{figure*}[t]
\includegraphics[scale=0.5]{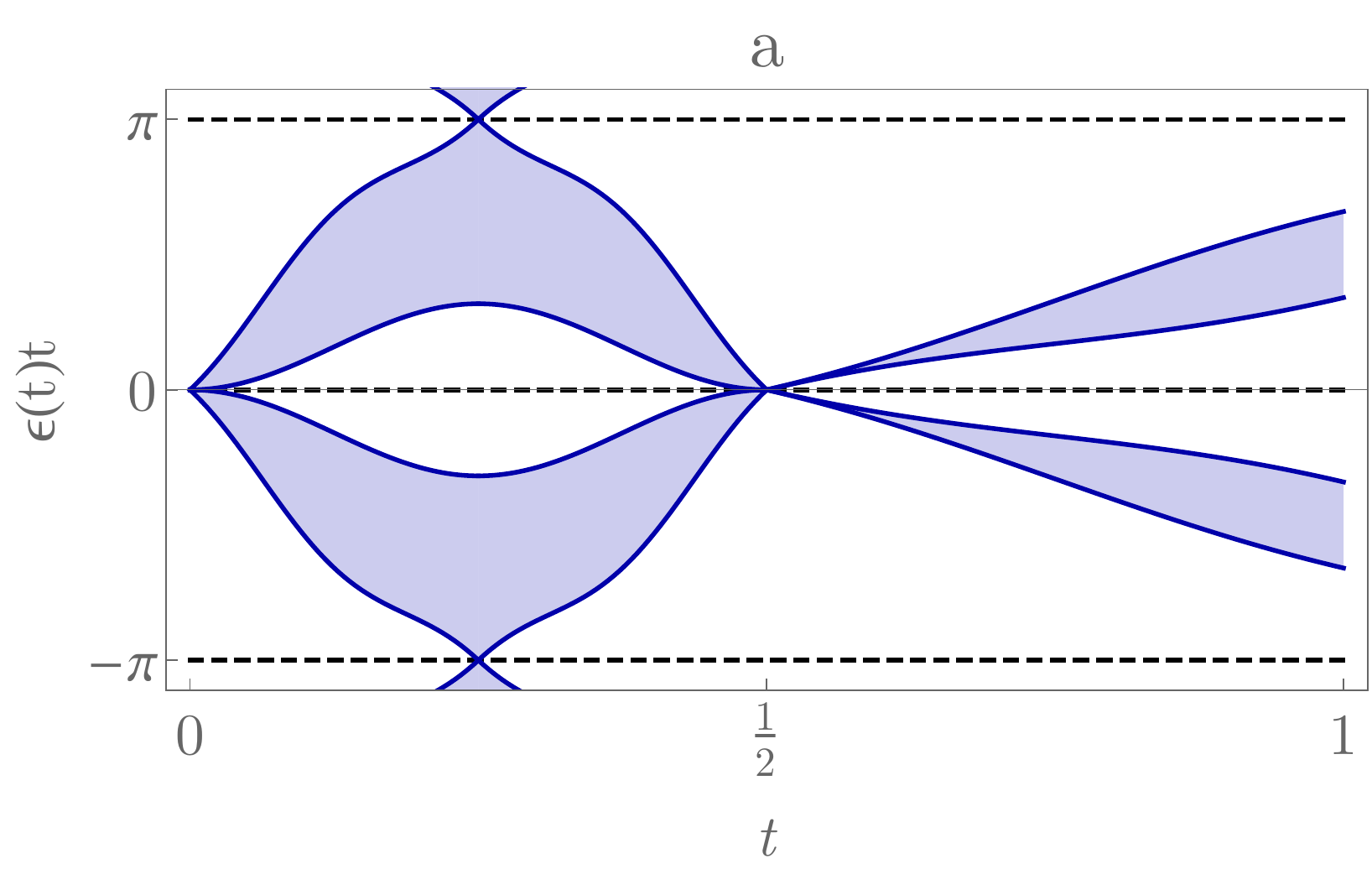}
\includegraphics[scale=0.5]{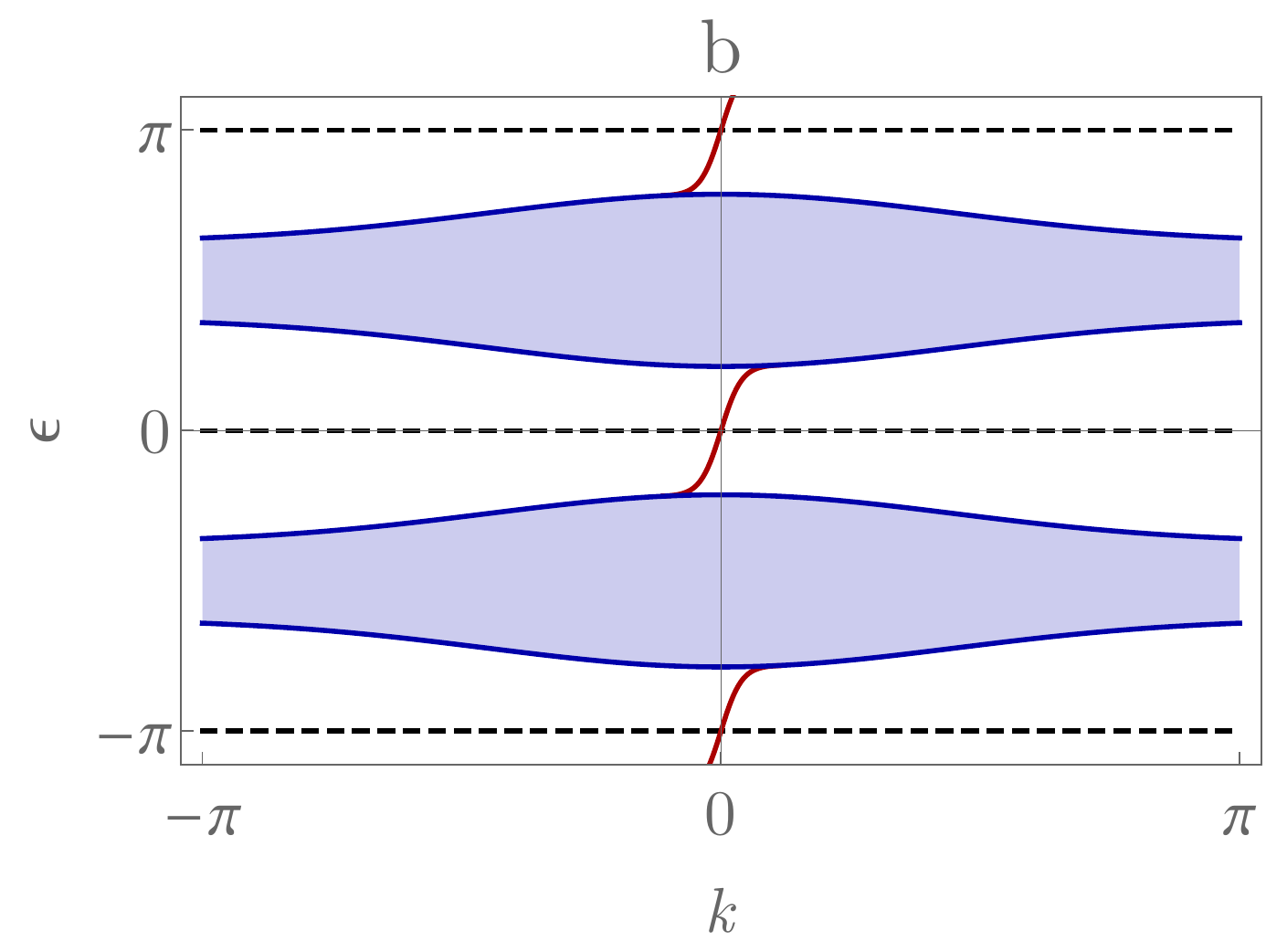}
\caption{(a) Unitary evolution as a composition of a loop with a constant Hamiltonian evolution. Instantaneous quasienergy bands are shown in blue. (b) End point of this unitary evolution, with quasienergy bands shown in blue and edge modes (which may be present in the open system) shown in red. The full spectrum has been projected onto a single momentum direction, labelled by $k$.\label{fig:gapped_unitary}}
\end{figure*}

The gapped spectrum in Fig.~\ref{fig:gapped_unitary}b resembles the band structure of a conventional, static Hamiltonian, with two well-separated bands and a gap at zero. In this situation, it is useful to define the effective Floquet Hamiltonian $H_F$ through
\beq
H_F(\bk)&=&i\ln\left[U(\bk,1)\right],\label{eq:FloquetHamDef}
\eeq
where the branch cut of the logarithm can be placed in the gap at $\epsilon=\pi$. According to Eq.~\eqref{eq:unitary_definition}, the Floquet Hamiltonian might naively be interpreted as the effective static Hamiltonian that, under time evolution, generates the quasienergy spectrum of the corresponding unitary, $U(\bk,1)$. If the Floquet Hamiltonian is topologically non-trivial, we might expect the edge modes associated with $H_F$ to transfer to edge modes in the quasienergy spectrum of $U(\bk,1)$. Indeed, if one considers evolution with a time-independent but topologically nontrivial Hamiltonian, then $U(\bk,1)$ for the open system will have robust edge modes if the corresponding unitary for the closed system is gapped at zero.

Although this intuition goes some way towards explaining the protected edge modes of unitary operators, the time-dependent situation is inherently more complicated: in general, there can be edge modes in the quasienergy gaps at both $\epsilon=0$ and $\epsilon=\pi$, the latter of which lie beyond a description in terms of the effective Floquet Hamiltonian. Indeed, recent studies have demonstrated systems that exhibit edge modes in both gaps even when the effective Floquet Hamiltonian is the identity operator (see, for example, Ref.~\onlinecite{Rudner:2013bg}). To fully characterize the edge modes of a unitary operator $U(\bk,1)$, we require information about the unitary evolution $U(\bk,t)$ throughout the period of evolution $0\leq t\leq 1$. In Fig.~\ref{fig:gapped_unitary}a we show a nontrivial evolution of this form that might generate edge modes in the gaps at $\epsilon=0$ and $\epsilon=\pi$. An interesting feature of Floquet systems with edge modes at $\epsilon=\pi$ is that the unitary for the open system cannot be written in the form $U(\bk,1)=e^{-iH(\bk)}$ for some local Hamiltonian $H(\bk)$. This contrasts with the closed system, whose unitary can always be written in this form.
\subsection{Compositions and Homotopy of Unitary Evolutions \label{sec:unitary_evolution_properties}}
Before outlining the classification scheme in detail, we describe a few properties of unitary evolutions that we will require below. We will need to consider compositions of unitaries, and so, borrowing notation from the composition of paths \cite{nakahara2003geometry}, we write the evolution due to $U_1$ followed by the evolution due to $U_2$ as $U_1* U_2$. If $H_1(\bk,t)$ is the Hamiltonian corresponding to $U_1$ and $H_2(\bk,t)$ is the Hamiltonian corresponding to $U_2$, the Hamiltonian corresponding to the composition $U_1* U_2$ is given by
\beq
H(t)&=&\left\{\renewcommand\arraystretch{1.2}\begin{array}{ccc}
H_1(\bk,2t) && 0\leq t\leq 1/2\\
H_2(\bk,2t-1) && 1/2\leq t\leq 1
\end{array}\right..\label{eq:composition}
\eeq
With this definition, the endpoint of any composition of unitaries always occurs at $t=1$.

In general, this composition rule produces an evolution that is no longer time-reversal symmetric, even if $H_1$ and $H_2$ individually \emph{are} time-reversal symmetric. If we wish to consider systems with TRS, we should instead define the Hamiltonian corresponding to the composition $U_1* U_2$ by
\beq
H(t)&=&\left\{\renewcommand\arraystretch{1.2}\begin{array}{ccc}
H_2(\bk,2t) && 0\leq t\leq 1/4\\
H_1(\bk,2t-1/2) && 1/4\leq t\leq 3/4\\
H_2(\bk,2t-1) && 3/4\leq t\leq 1
\end{array}\right.,\label{eq:composition_TRS}
\eeq
which we see has the required symmetry.

For classification purposes, we split the set $\mathcal{U}^\mathcal{S}_g$ into equivalence classes. Following the classification scheme of static topological insulators in Ref.~\onlinecite{Kitaev:2009vc}, we carry out this partition using homotopy. We define homotopy in the usual way, and say that two unitary operators $U_1,U_2\in\mathcal{U}^\mathcal{S}_g$ are homotopic if and only if there exists a function $h(s)$, with $s\in\left[0,1\right]$, such that
\beq
h(0)=U_1,&~~~~~~&h(1)=U_2,
\eeq
with $h(s)\in\mathcal{U}^\mathcal{S}_g$ for all intermediate values of $s$. In this way, the gap structure at $t=1$ (but only the gap structure at this point) must be preserved throughout the homotopy. We write homotopy equivalence as $U_1\approx U_2$.

In order to compare unitaries with different numbers of bands, we introduce the further equivalence relation of stable homotopy as follows. We define $U_1\sim U_2$ if and only if there exist two trivial unitaries, $U_{n_1}^0$ and $U_{n_2}^0$, such that
\beq
U_1\oplus U^0_{n_1}&\approx& U_2\oplus U_{n_2}^0,
\eeq
where $\oplus$ is the direct sum and $n_1,n_2$ are positive integers that give the number of bands in the trivial unitary. The appropriate trivial unitaries $U^0_n$ must belong to the set $\mathcal{U}^\mathcal{S}_{g}$, and will be given explicitly when required.

Finally, since we are ultimately interested in the behavior at a system boundary, the discussion is simplified considerably if we also define equivalence classes of pairs of unitaries. The pairs $(U_1,U_2)$ and $(U_3,U_4)$, where both members of each pair have the same number of bands, are stably homotopic if and only if
\beq
U_1\oplus U_4&\sim&U_2\oplus U_3.
\eeq
We write this equivalence as $(U_1,U_2)\sim(U_3,U_4)$.

\section{Decomposition of Unitary Evolutions\label{sec:unitary_decomposition}}
Our approach will be to isolate the new, dynamical topological behavior from the static topological behavior that is encoded in a nontrivial Floquet Hamiltonian. We will initially restrict the discussion to unitaries that have gaps at both $\epsilon=0$ and $\epsilon=\pi$ (with possible additional gaps elsewhere in the spectrum), considering more general cases in Sec.~\ref{sec:discussion}. We write the set of such unitaries as $\mathcal{U}^\mathcal{S}_{0,\pi}$. 

To proceed, it is useful to define two special types of unitary evolution. First, we define a unitary loop to be a unitary evolution that satisfies $U(\bk,0)=U(\bk,1)=\mathbb{I}$. A unitary of this form can be seen to act trivially on a closed system, but may generate nontrivial edge modes in a system with a boundary. Secondly, we define a constant Hamiltonian evolution as a unitary evolution that may be expressed as $U(\bk,t)=e^{-iH(\bk)t}$ for some static Hamiltonian $H(\bk)$, whose eigenvalues have magnitude strictly less than $\pi$. The utility of  identifying these two types of unitary evolutions becomes apparent when one considers the following theorem:
\begin{theorem}\label{thm:decomposition}
Every unitary $U\in\mathcal{U}^\mathcal{S}_{0,\pi}$ can be continuously deformed to a composition of a unitary loop $L$ and a constant Hamiltonian evolution $C$, which we write as $U\approx L*C$. $L$ and $C$ are unique up to homotopy.
\end{theorem}
Theorem~\ref{thm:decomposition} is proved in Appendix~\ref{sec:decomposition_proof} and is illustrated schematically in Fig.~\ref{fig:gapped_unitary}a. By a slight abuse of terminology, we will often call the unitary composed of the loop and the constant the `decomposition' of the original unitary. 
Heuristically, the decomposition can be interpreted as an initial loop, which may generate nontrivial edge modes at $\epsilon=\pi$, followed by a constant evolution by the static Floquet Hamiltonian $H_F$. Since we have assumed there is a spectral gap at $\epsilon=\pi$, the branch cut required for the definition in Eq.~\eqref{eq:FloquetHamDef} can be placed in this region, and the final quasienergy bands can be consistently thought of as emanating from the point $\epsilon=0$. In addition, since we are assuming that the complete unitary evolution is gapped at both $\epsilon=0$ and $\epsilon=\pi$, the static Hamiltonian required for the constant evolution must be gapped at zero. We write the set of unitary loops in symmetry class $\mathcal{S}$ as $\mathcal{U}^\mathcal{S}_L$ and the set of constant, gapped Floquet Hamiltonian evolutions in symmetry class $\mathcal{S}$ as $\mathcal{U}^\mathcal{S}_C$.

Through this unique decomposition, we see that a general unitary evolution from $\mathcal{U}_{0,\pi}$ can be classified by separately considering the unitary loop component and the constant evolution component. A specific phase may be labelled by the pair $(n_L, n_C)$, where $n_L$ and $n_C$ are invariant integers associated with the unitary loop and constant evolution components, respectively. 

Next, we label the set of static gapped Hamiltonians in symmetry class $\mathcal{S}$, whose eigenvalues $E$ satisfy $0<|E|<\pi$, by $\mathcal{H}^\mcs$. The set of gapped Floquet Hamiltonian evolutions in $\mathcal{U}^\mathcal{S}_C$ is clearly in one-to-one correspondence with the set of static Hamiltonians in $\mathcal{H}^\mcs$. This follows from the bijection $C(t)=\exp\left(-iH_Ft\right)$, where $H_F$ is the unique Floquet Hamiltonian with eigenvalue magnitude strictly between $0$ and $\pi$. From the definition of homotopy given in Sec.~\ref{sec:unitary_evolution_properties}, we see that $C_1\approx C_2$ within $\mathcal{U}^\mathcal{S}_C$ if and only if $H_1\approx H_2$ within $\mathcal{H}^\mcs$, where $H_i$ is the static Hamiltonian corresponding to $C_i$. In addition, it follows that $C_1\sim C_2$ if and only if $H_1\sim H_2$, if we write the trivial unitary as
\beq
U^0_n(\bk,t)=\exp\left(-iH_n^0\right),
\eeq
where $H^0_n$ is a suitable trivial Hamiltonian.

In this way, we can classify pairs of constant Hamiltonian evolutions $(C_1,C_2)$ by instead classifying pairs of static Hamiltonians $(H_1,H_2)$. This is a relative classification that is equivalent to the well-known classification of static topological insulators, which is summarized in the periodic table given in Ref.~\onlinecite{Kitaev:2009vc}. The classification of pairs of unitary loops $(L_1,L_2)$ does not, however, have a static analogue.

Through this decomposition, the periodic table of static topological insulators may be viewed as a subset of a larger classification scheme that also includes time-dependent Hamiltonians. In this picture, static topological insulators correspond to compositions of nontrivial constant Hamiltonian evolutions with trivial unitary loops. More general dynamical topological phases arise through compositions of constant evolutions with \emph{nontrivial} unitary loops. In the next section of this paper we set out to classify the nontrivial unitary loops that may exist in each symmetry class.

\section{Classification of Unitary Loops\label{sec:loop_classification}}
\subsection{Unitary Loops with Particle-Hole Symmetry Only\label{sec:PHS_classification}}

With the machinery defined in previous sections, we are now ready to give a systematic discussion of the classification of unitary loops. We begin by considering loops in systems that have PHS but no other symmetry, belonging to the set $\mathcal{U}^\mathcal{S}_L$ with $\mathcal{S}\in\{\rm C,D\}$.
\subsubsection{Hermitian Maps}
To proceed, we define a Hermitian map corresponding to a given unitary $U(\bk,t)$ through
\beq
H_U(\bk,t)&=&\left(\begin{array}{cc}
0 & U(\bk,t)\\
U^\dagger(\bk,t) & 0
\end{array}\right),\label{eq:hermitian_map_definition}
\eeq
which we see satisfies $H_U^2=\mathbb{I}$. In addition, we define the two new symmetry operations
\beq
P_{1}=\left(\begin{array}{cc}
P & 0\\
0 & P
\end{array}\right),&~~~~~~&P_{2}=\left(\begin{array}{cc}
P & 0\\
0 & -P
\end{array}\right),
\eeq
which are derived from the standard PHS operator $P$. Using Eq.~\eqref{eq:unitary_PHS}, we see that the Hermitian map satisfies the following new symmetry relations:
\beq
P_{1}H_U(\bk,t)P_{1}^{-1}&=&H_U^*(-\bk,t)\nonumber\\
P_{2}H_U(\bk,t)P_{2}^{-1}&=&-H_U^*(-\bk,t).
\eeq
We write the set of Hermitian maps that square to the identity, satisfy these symmetries and which additionally satisfy $H_U(\bk,0)=H_U(\bk,1)$ (but which are not necessarily derived from unitary loops) as $\mathscr{H}^\mathcal{S}$. We write the subset of $\mathscr{H}^\mathcal{S}$ that corresponds specifically to unitary loops as $\mathscr{H}^\mathcal{S}_L$, and note that from the properties of unitary loops, members of $\mathscr{H}^\mathcal{S}_L$ must satisfy
\beq
H_U(\bk,0)=H_U(\bk,1)&=&\left(\begin{array}{cc}
0 & \mathbb{I}_{n}\\
\mathbb{I}_{n} & 0
\end{array}\right).
\eeq
There is a one-to-one mapping between a unitary loop $U\in\mathcal{U}^\mathcal{S}_L$ and the corresponding Hermitian map $H_U\in\mathscr{H}^\mathcal{S}_L$, a statement that is proved in Appendix~\ref{sec:121_proof}.

It is easy to verify that $U_1\approx U_2$ if and only if $H_{U_1}\approx H_{U_2}$, which extends the definition of homotopy equivalence to $\mathscr{H}^\mathcal{S}$. Next, we note that the trivial unitary loop is given by
\beq
U_n^0(\bk,t)&=&\mathbb{I}_n,
\eeq
and the corresponding trivial matrix in $\mathscr{H}^\mathcal{S}_L$ is given by
\beq
H_{U,n}^0(\bk,t)&=&\left(\begin{array}{cc}
0 & \mathbb{I}_{n}\\
\mathbb{I}_{n} & 0
\end{array}\right).
\eeq
This allows us to define the stable homotopy equivalence of Hermitian maps through
\beq
H_A\sim H_B &~~\Leftrightarrow~~& H_A\oplus H_{U,n_1}^0\approx H_B\oplus H_{U,n_2}^0,
\eeq
in the space $\mathscr{H}^\mathcal{S}_L$. Again, it is clear that $U_1\sim U_2\Leftrightarrow H_{U_1}\sim H_{U_2}$. 

As in the case of unitaries, we can also consider pairs of Hermitian maps, $(H_{U_1},H_{U_2})$, where both members of each pair have the same number of bands. This allows us to define the equivalence relation $(H_{U_1},H_{U_2})\sim(H_{U_3},H_{U_4})$ if and only if $H_{U_1}\oplus H_{U_4}\sim H_{U_3}\oplus H_{U_2}$. These pairs of Hermitian maps form an additive group, described in Appendix~\ref{app:KTheory}, which we can use to classify the relative topological invariants of the corresponding pair of unitary evolutions.
\subsubsection{Classification of Unitaries using K-Theory}
We will omit the technical steps of the K-theory argument in this section, and instead give an overview of the method. For further details, we refer the reader to Appendix~\ref{app:KTheory} and references therein.

The general idea is to use the Morita equivalence of categories to map the group of equivalence classes of pairs in $\mathscr{H}^\mathcal{S}$ onto a K-group of the kind $KR^{0,q}(M)$ (or, later, $K(M)$ for classes A and AIII). The $KR^{0,q}(M)$ are a set of well-studied K-groups of manifolds which are described, for example, in Refs.~\onlinecite{Atiyah:112984,Karoubi:2023223,Atiyah:1966hg}. In these expressions, $M$ is the manifold $S^{1} \times X$, where $X$ is the Brillouin zone and $S^{1}$ is the circle corresponding to the time direction, whose initial and final points ($t=0$ and $t=1$) are identified due to the assumed periodicity of Hermitian maps in $\mathscr{H}^\mathcal{S}$. The space $M$ is, in the terminology of Ref.~\onlinecite{Atiyah:1966hg}, a \textit{real space}, i.e. a space with an involution corresponding to $\bk \rightarrow - \bk$.
The reduction using Morita equivalence relations is equivalent to Kitaev's trick of replacing negative generators with positive generators \cite{Kitaev:2009vc}. 

For class D, the resulting group of the equivalence classes of pairs is $KR^{0,1}(S^1\times X)$, while for class C the resulting group is $KR^{0,5}(S^1\times X)$. Specifically restricting ourselves to the subset $\mathscr{H}^\mathcal{S}_L$, the group of equivalence classes of pairs of loops is then isomorphic to the relative K-group
\begin{widetext}
\begin{equation}
\begin{array}{ccc}
KR^{0,1}(S^{1} \times X,\{0\}\times X)&=&KR^{0,2}(X)\mbox{~~~~~~Class D}\\
KR^{0,5}(S^{1} \times X,\{0\}\times X)&=&KR^{0,6}(X)\mbox{~~~~~~Class C},
\end{array}\label{eq:PHS_kgroups}
\end{equation}
\end{widetext}
where the point $\{0\} \in S^{1}$ corresponds to the initial time of the evolution. The equalities in these two equations  are well-known K-theory isomorphisms \cite{Karoubi:2023223,Atiyah:1966hg}. The last results are identical to the K-groups classifying static topological insulators from these classes, and we note that the K-group captures both the strong and weak invariants.

\subsection{Unitary Loops with Time-reversal Symmetry}
We now discuss the classification of unitaries that have TRS, and which may also have PHS. These correspond to the symmetry classes AI and AII (TRS only), and classes BDI, CII, DIII and CI (TRS and PHS). 

Although it is possible to work with the unitary operators directly, the calculations become considerably simpler if we instead define symmetrized unitaries, $U_S(\bk,t)$, through 
\beq
U_S(\bk,t)&=&\exp\left[-i\int^{\frac{1+t}{2}}_{\frac{1-t}{2}}H(\bk,t')\,\dd t'\right].\label{eq:modifiedU_def}
\eeq
It is clear that there is a one-to-one correspondence between unitary operators $U(\bk,t)$ and their symmetrized forms $U_S(\bk,t)$, and further, that both expressions agree at $t=0$ and $t=1$. Under a particle-hole transformation, a symmetrized unitary with PHS satisfies the same relation as the original unitary,
\beq
PU_S(\bk,t)P^{-1}&=&U^{*}_S(-\bk,t),
\eeq
while under time-reversal, the symmetrized unitary operator transforms as
\beq
\theta U_S(\bk,t)\theta^{-1}=U^{\dagger*}_S(-\bk,t),\label{eq:modifiedU_TRS}
\eeq
relations that are derived in Appendix~\ref{sec:modified_symmetry_actions}. For the rest of this section we will drop the subscript $S$ and assume that we are using symmetrized unitaries.

As in the previous section, a (symmetrized) unitary evolution that belongs to $\mathcal{U}^\mathcal{S}_{0,\pi}$ is equivalent to a composition of a unitary loop with a constant Hamiltonian evolution. However, since the unitaries involved now have TRS, composition is defined using the time-reversal symmetric expression in Eq.~\eqref{eq:composition_TRS}. The classification of the constant Hamiltonian evolution component follows the discussion in Sec.~\ref{sec:unitary_decomposition}, with topological edge modes at $\epsilon=\pi$, if present, arising from the loop component.

\subsubsection{Hermitian Maps\label{sec:loops_PHS_TRS}}
To classify the unitary loops in these classes, we again define a Hermitian map corresponding to a given (symmetrized) unitary $U(\bk,t)$ as in Eq.~\eqref{eq:hermitian_map_definition}.
This time, we require up to four symmetry operators,
\beq
P_{1}=\left(\begin{array}{cc}
P & 0\\
0 & P
\end{array}\right),&~~~~~~&P_{2}=\left(\begin{array}{cc}
P & 0\\
0 & -P
\end{array}\right),\nonumber\\
\theta_{1}=\left(\begin{array}{cc}
0 & \theta\\
\theta & 0
\end{array}\right),&~~~~~~&\theta_{2}=\left(\begin{array}{cc}
0 & \theta\\
-\theta & 0
\end{array}\right),
\eeq
which are derived from the symmetry operators $P$ and $\theta$. If the relevant symmetry is present, these operators act on the Hermitian map $H_U$ to give
\beq
P_{1}H_U(\bk,t)P_{1}^{-1}&=&H_U^*(-\bk,t)\nonumber\\
P_{2}H_U(\bk,t)P_{2}^{-1}&=&-H_U^*(-\bk,t)
\eeq
for classes BDI, CII, DIII and CI, and 
\beq
\theta_{1}H_U(\bk,t)\theta_{1}^{-1}&=&H_U^*(-\bk,t)\nonumber\\
\theta_{2}H_U(\bk,t)\theta_{2}^{-1}&=&-H_U^*(-\bk,t)
\eeq
for classes AI, AII, BDI, CII, DIII and CI.

As before, we write the set of Hermitian maps that square to the identity, satisfy these symmetries, and which additionally satisfy $H_U(\bk,0)=H_U(\bk,1)$, as $\mathscr{H}^\mathcal{S}$, and write the subset of this that corresponds to unitary loops as $\mathscr{H}^\mathcal{S}_L$. There is again a one-to-one mapping between the set of $U\in\mathcal{U}^\mathcal{S}_L$ and the corresponding set of Hermitian maps $H_U\in\mathscr{H}^\mathcal{S}_L$, a statement that can be proved using a method similar to that given in Appendix~\ref{sec:121_proof}. As in Sec.~\ref{sec:PHS_classification}, homotopy, stable homotopy, and the equivalence of pairs can be defined for Hermitian maps in $\mathscr{H}^\mathcal{S}$.
\subsubsection{Classification of Unitaries using K-Theory}
We can now use K-theory arguments to map the equivalence classes of pairs in $\mathscr{H}^\mathcal{S}$ onto K-groups. Using the arguments of Sec.~\ref{sec:PHS_classification} for each symmetry class, we find the group of equivalence classes in each case maps onto
\begin{widetext}
\beq
\begin{array}{ccccc}
KR^{0,7}(S^1\times X)&~~~~~~&\mbox{Class AI}\\
KR^{0,3}(S^1\times X)&~~~~~~&\mbox{Class AII}\\
KR^{0,0}(S^1\times X)&~~~~~~&\mbox{Class BDI}\\
KR^{0,4}(S^1\times X)&~~~~~~&\mbox{Class CII}\\
KR^{0,2}(S^1\times X)&~~~~~~&\mbox{Class DIII}\\
KR^{0,6}(S^1\times X)&~~~~~~&\mbox{Class CI}.
\end{array}
\eeq
Restricting to the subset $\mathscr{H}^\mathcal{S}_L$, the groups are then isomorphic to the relative K-groups
\beq
\begin{array}{ccccc}
KR^{0,7}(S^1\times X,\{0\}\times X)&=&KR^{0,0}(X)\mbox{~~~~~~Class AI}\\
KR^{0,3}(S^1\times X,\{0\}\times X)&=&KR^{0,4}(X)\mbox{~~~~~~Class AII}\\
KR^{0,0}(S^1\times X,\{0\}\times X)&=&KR^{0,1}(X)\mbox{~~~~~~Class BDI}\\
KR^{0,4}(S^1\times X,\{0\}\times X)&=&KR^{0,5}(X)\mbox{~~~~~~Class CII}\\
KR^{0,2}(S^1\times X,\{0\}\times X)&=&KR^{0,3}(X)\mbox{~~~~~~Class DIII}\\
KR^{0,6}(S^1\times X,\{0\}\times X)&=&KR^{0,7}(X)\mbox{~~~~~~Class CI},
\end{array}\label{eq:PHS_TRS_kgroups}
\eeq
\end{widetext}
using a set of well-known K-theory isomorphisms as outlined in Appendix E~\cite{Karoubi:2023223,Atiyah:1966hg}. The last results are identical to the K-groups classifying static topological insulators from these classes and describe the complete set of strong and weak invariants. Overall, it follows that pairs of unitary loops within $\mathcal{U}^\mathcal{S}_L$ are classified by the K-groups given in Eq.~\eqref{eq:PHS_TRS_kgroups}.

\subsection{Classification of Gapped Unitaries in Symmetry Classes A and AIII}
Finally, we discuss the classification of time evolution unitaries in the complex symmetry classes, with $\mathcal{S}\in\{\rm A,AIII\}$. As in the previous section, the discussion is simplified if we use the symmetrized unitaries $U_S(\bk,t)$ defined in Eq.~\eqref{eq:modifiedU_def}. In terms of these, the chiral symmetry operator (relevant for class AIII) has the action
\beq
C U_S(\bk,t)C^{-1}=U_S^{\dagger}(\bk,t),
\eeq
a relation that is derived in Appendix~\ref{sec:modified_symmetry_actions}. As before, we will drop the subscript $S$ and assume we are working with symmetrized unitaries throughout this section.

\subsubsection{Hermitian Maps\label{sec:loops_A_AIII}}
As in the previous cases, we use Eq.~\eqref{eq:hermitian_map_definition} to define a Hermitian map $H_U(\bk,t)$ (satisfying $H_U^2=\mathbb{I}$), which corresponds to a given (symmetrized) unitary $U(\bk,t)$. The relevant symmetry operators for classes A and AIII are
\beq
\Sigma=\left(\begin{array}{cc}
\mathbb{I} & 0\\
0 & -\mathbb{I}
\end{array}\right),&~~~~~~&\Gamma=\left(\begin{array}{cc}
0 & C\\
-C & 0
\end{array}\right).
\eeq
The first of these anticommutes with any Hermitian map of the form $H_U$, while the second, which is derived from the CS operator $C$, is relevant only for class AIII. These operators act on $H_U$ to give
\beq
\Sigma H_U(\bk,t)\Sigma ^{-1}=-H_U(\bk,t)&~~~~~~&\mbox{Classes A and AIII}\nonumber\\
\Gamma H_U(\bk,t)\Gamma^{-1}=-H_U(\bk,t)&~~~~~~&\mbox{Class AIII}.
\eeq
We write the set of Hermitian maps that square to the identity, satisfy these symmetries, and which additionally satisfy $H_U(\bk,0)=H_U(\bk,1)$ as $\mathscr{H}^\mathcal{S}$, and write the subset of this that corresponds to unitary loops as $\mathscr{H}^\mathcal{S}_L$. There is again a one-to-one mapping between the set of $U\in\mathcal{U}^\mathcal{S}_L$ and the corresponding set of Hermitian maps $H_U\in\mathscr{H}^\mathcal{S}_L$, a statement that can be proved using the method of Appendix~\ref{sec:121_proof}. Homotopy, stable homotopy and equivalence of pairs in $\mathscr{H}^\mathcal{S}$ can be defined as in previous sections.
\subsubsection{Classification of Unitaries using K-Theory}
We can now use K-theory arguments to map the equivalence classes of pairs in $\mathscr{H}^\mathcal{S}$ onto K-groups. For each symmetry class, we find the mapping to
\beq
\begin{array}{ccccc}
K^{1}(S^{1}\times X)&~~~~~~&\mbox{Class A}\\
K^{2}(S^{1}\times X)&~~~~~~&\mbox{Class AIII}.
\end{array}
\eeq
Restricting to the subset $\mathscr{H}^\mathcal{S}_L$, the groups are then isomorphic to the relative K-groups
\begin{widetext}
\beq
\begin{array}{ccccc}
K^{1}(S^{1}\times X,\{0\}\times X)=K^{0}(X)&&\mbox{~~~~~~Class A}\\
K^{2}(S^{1}\times X,\{0\}\times X)=K^{1}(X)&&\mbox{~~~~~~Class AIII},
\end{array}\label{eq:A_AIII_kgroups}
\eeq
\end{widetext}
which follow from known K-theory isomorphisms \cite{Karoubi:2023223,Atiyah:1966hg}. The last results are identical to the K-groups classifying static topological insulators from these classes, and it follows overall that pairs of unitary loops from classes A and AIII are classified by the K-groups given in Eq.~\eqref{eq:A_AIII_kgroups}. The K-groups capture the complete set of strong and weak invariants.

\begin{table*}[t!]
\be
\arraycolsep=7pt\renewcommand\arraystretch{0.6}\begin{array}{c|cccccccc}
\mathcal{S} & d=0 & 1 & 2 & 3 & 4 & 5 & 6 & 7 \\
\hline
{\rm A}& \mathbb{Z}\times\mathbb{Z}& \emptyset & \mathbb{Z}\times\mathbb{Z}& \emptyset & \mathbb{Z}\times\mathbb{Z}& \emptyset & \mathbb{Z}\times\mathbb{Z}& \emptyset \\
{\rm AIII}  & \emptyset & \mathbb{Z}\times\mathbb{Z}& \emptyset & \mathbb{Z}\times\mathbb{Z}& \emptyset & \mathbb{Z}\times\mathbb{Z}& \emptyset & \mathbb{Z}\times\mathbb{Z}\\
\hline
{\rm AI} & \mathbb{Z}\times\mathbb{Z}& \emptyset & \emptyset & \emptyset & \mathbb{Z}\times\mathbb{Z}& \emptyset & \mathbb{Z}_2\times\mathbb{Z}_2 & \mathbb{Z}_2\times\mathbb{Z}_2 \\
{\rm BDI} & \mathbb{Z}_2\times\mathbb{Z}_2 & \mathbb{Z}\times\mathbb{Z}& \emptyset & \emptyset & \emptyset & \mathbb{Z}\times\mathbb{Z}& \emptyset & \mathbb{Z}_2\times\mathbb{Z}_2\\
{\rm D} & \mathbb{Z}_2\times\mathbb{Z}_2 & \mathbb{Z}_2\times\mathbb{Z}_2 & \mathbb{Z}\times\mathbb{Z}& \emptyset & \emptyset & \emptyset & \mathbb{Z}\times\mathbb{Z}& \emptyset \\
{\rm DIII} & \emptyset & \mathbb{Z}_2\times\mathbb{Z}_2 & \mathbb{Z}_2\times\mathbb{Z}_2 & \mathbb{Z}\times\mathbb{Z}& \emptyset & \emptyset & \emptyset & \mathbb{Z}\times\mathbb{Z}\\
{\rm AII} & \mathbb{Z}\times\mathbb{Z}& \emptyset & \mathbb{Z}_2\times\mathbb{Z}_2 & \mathbb{Z}_2\times\mathbb{Z}_2 & \mathbb{Z}\times\mathbb{Z}& \emptyset & \emptyset & \emptyset \\
{\rm CII} & \emptyset & \mathbb{Z}\times\mathbb{Z}& \emptyset & \mathbb{Z}_2\times\mathbb{Z}_2 & \mathbb{Z}_2\times\mathbb{Z}_2 & \mathbb{Z}\times\mathbb{Z}& \emptyset & \emptyset \\
{\rm C} & \emptyset & \emptyset & \mathbb{Z}\times\mathbb{Z}& \emptyset & \mathbb{Z}_2\times\mathbb{Z}_2 & \mathbb{Z}_2\times\mathbb{Z}_2 & \mathbb{Z}\times\mathbb{Z}& \emptyset \\
{\rm CI} & \emptyset & \emptyset & \emptyset & \mathbb{Z}\times\mathbb{Z}& \emptyset & \mathbb{Z}_2\times\mathbb{Z}_2 & \mathbb{Z}_2\times\mathbb{Z}_2 & \mathbb{Z}\times\mathbb{Z}\\
\end{array}.
\ee
\caption{Classification of two-gapped unitaries by symmetry class and spatial dimension $d$. The table repeats for $d\geq 8$ (Bott periodicity). \label{tab:AZ_class_2g}}
\end{table*}

\section{Discussion\label{sec:discussion}}
In the preceding section we obtained the groups of the equivalence classes of pairs of unitary loops from the ten AZ symmetry classes. These groups are of the form $KR^{0,q}(X)$ for real symmetry classes and  of the form $K^q(X)$ for complex symmetry classes, where $X$ is the Brillouin zone torus. The final K-groups were given in Eqs.~(\ref{eq:PHS_kgroups},\ref{eq:PHS_TRS_kgroups},\ref{eq:A_AIII_kgroups}).

We noted that these K-groups are identical to those obtained from the classification of static (single-gapped) topological Hamiltonians in the same symmetry classes. Depending on the dimension of the Brillouin zone $X$, these K-groups are isomorphic to a group $G\in\{\emptyset,\mathbb{Z}_2,\mathbb{Z}\}$, reproducing the well-known periodic table of topological insulators and superconductors \cite{Kitaev:2009vc}.

If a general unitary evolution can be homotopically deformed to a loop (for example, if there is only one spectral gap at $\epsilon=\pi$), then it falls within the loop classification scheme discussed above. In Sec~\ref{sec:unitary_decomposition}, however, we explained that a unitary loop is just one component of a generic unitary evolution. More commonly, a unitary evolution leads to a final unitary with a set of spectral gaps  at $\epsilon=0$ and $\epsilon=\pi$ (and possible gaps elsewhere in the spectrum), in which case it can be continuously deformed  to a loop followed by a constant Hamiltonian evolution, which itself may be topologically nontrivial. We argued in Sec.~\ref{sec:unitary_decomposition} that the constant evolution, being in one-to-one correspondence with a static Hamiltonian, also follows the usual classification scheme for static topological insulators. For a static Hamiltonian with a single gap at zero, this introduces an additional factor of $KR^{0,q}(X)$ for real symmetry classes and an additional factor of $K^q(X)$ for complex symmetry classes. Combining both the loop and constant Hamiltonian components, we see that the equivalence classes of pairs of unitary evolutions from $\mathcal{U}_{0,\pi}$ are isomorphic to a product of the form $KR^{0,q}(X)\times KR^{0,q}(X)$ (real symmetry classes) or $K^q(X)\times K^q(X)$ (complex symmetry classes). Depending on the dimension of the Brillouin zone, these products are isomorphic to a group $G\times G\in\{\emptyset,\mathbb{Z}_2\times \mathbb{Z}_2,\mathbb{Z}\times\mathbb{Z}\}$, as shown in Table~\ref{tab:AZ_class_2g}. 

In this way, the periodic table for static topological insulators is contained within the periodic table shown in Table.~\ref{tab:AZ_class_2g}. Static topological insulators correspond to evolutions with a trivial unitary loop component, which, in our generalized classification scheme, leads to one factor of the classifying group $G\times G$ being trivial.

At the interface between a system described by unitary $U_1$ and a system described by unitary $U_2$, the principle of bulk-edge correspondence asserts that there should exist protected edge modes, shown schematically in Fig.~\ref{fig:bulk_edge}. A particular edge mode can be labelled by a quantum number, and the complete set of quantum numbers is isomorphic to the set of equivalence classes of $(U_1,U_2)$. We can therefore determine the quantum number of the edge modes by appealing to the bulk classification scheme outlined above.

\begin{figure}[t]
\includegraphics[scale = 0.5]{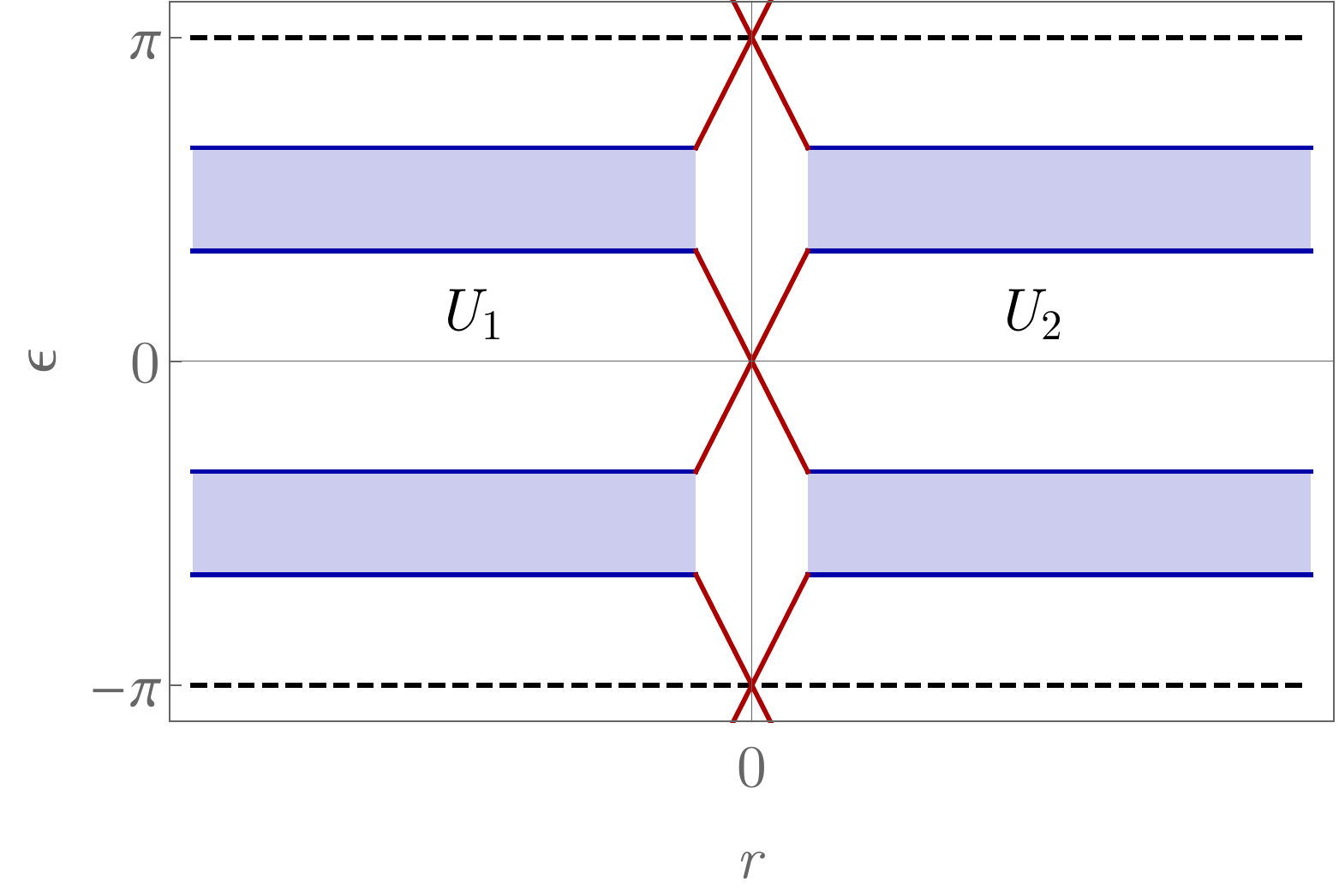}
\caption{Schematic diagram of the interface between a system described by unitary $U_1$ and a system described by unitary $U_2\neq U_1$. The vertical axis shows the quasienergy spectrum at $t=T$ and the horizontal axis gives the displacement, with the interface occurring in the neighborhood around $r=0$. Bulk bands are shown in blue, with protected edge modes shown in red at the interface.\label{fig:bulk_edge}}
\end{figure}

To simplify the discussion, we again consider only unitary evolutions that are gapped at both $\epsilon=0$ and $\epsilon=\pi$. This is generic enough to exhibit Floquet topological edge modes, and the conclusions can easily be extended to systems with additional gaps. We note that a unitary evolution of this form is classified by a pair of integers from the appropriate group $G\times G$, according to Table~\ref{tab:AZ_class_2g}. We write the integer associated with the loop component as $n_L$ and the integer associated with the gap in the constant Hamiltonian evolution as $n_C$. These must have a one-to-one relation with the quantum numbers associated with the edge modes in the gaps, which we write as $n_0$ for the gap at $\epsilon=0$, and as $n_\pi$ for the gap at $\epsilon=\pi$.

We leave a full discussion of the bulk-edge correspondence to future work, but for completeness, we note here that the two sets of integers are related through
\beq
n_\pi&=&n_L\nonumber\\
n_0&=&n_C+n_L,\label{eq:bulk_edge_integers}
\eeq
where addition is taken modulo 2 for systems with a $\mathbb{Z}_2$ classification. In general, we see that the edge modes associated with the gap at $\epsilon=0$ and those at $\epsilon=\pi$ may be different. 

In a system with PHS, the quasienergy spectrum must be symmetric about $\epsilon=0$ and $\epsilon=\pi$. If there are gaps at these points, then Majorana modes may be present at the boundary of an open system. However, for a system without PHS, there is nothing special about the quasienergies $\epsilon=0$ and $\epsilon=\pi$. In these cases, we can always homotopically deform the unitary evolution so that the gaps occur at any values of our choosing, while staying within the symmetry class. In this way, for a unitary with two gaps \emph{anywhere} in the spectrum and without PHS, we can use homotopy to move the gaps to $\epsilon=0$ and $\epsilon=\pi$ so that the correspondence given by Eq.~\eqref{eq:bulk_edge_integers} continues to hold. More generally, for a system with $n_g$ gaps and no PHS, there will be an invariant integer associated with each gap in the Floquet Hamiltonian, $n_{C_i}$ with $1\leq i < n_g$. A simple extension of Eq.~\eqref{eq:bulk_edge_integers} gives the edge invariant of the $i$th gap as $n_i=n_{C_i}+n_L$. 

\section{Conclusions\label{sec:conclusions}}
In this paper we have used methods from K-theory to systematically classify noninteracting Floquet topological insulators across all AZ symmetry classes and dimensions. In the process, we discover a number of new topological Floquet phases. It would be interesting to see if these can be realized in an experimental setting. Our results, summarized in Table~\ref{tab:AZ_class_2g}, show that the classification of a static topological system described by group $G$ is extended to the product group $G\times G$ in the time-dependent case, assuming a canonical quasienergy spectrum with gaps at $\epsilon=0$ and $\epsilon=\pi$. Our approach uses the fact that a general time-evolution operator can be continuously deformed into a unitary loop followed by a constant Hamiltonian evolution, and the two factors in the resulting classification scheme can be interpreted as arising from these two unitary components. In Sec.~\ref{sec:discussion} we stated how this bulk classification scheme relates to the number of protected edge modes that may arise in a system with a boundary. The discussion of topological invariants and the bulk-boundary correspondence for more general band structures are interesting avenues for future work. 

As noted in the introduction, some elements of our periodic table have appeared elsewhere in the literature in the context of Floquet systems, using different methods \cite{Jiang:2011cw,Thakurathi:2013dt,Rudner:2013bg,Carpentier:2015dn,Asboth:2014bg,Nathan:2015bi,Fruchart:2016hk}. While our results are consistent with these works, a detailed comparison yields a number of differences. First, our definition of chiral symmetry differs from that of Ref.~\onlinecite{Asboth:2014bg}, and it would be worth investigating under what circumstances these definitions are equivalent and to what extent this affects the classification scheme. Secondly, Ref.~\onlinecite{Rudner:2013bg} introduces a frequency domain formulation for the study of Floquet systems that explicitly makes use of time periodicity. It would be of interest to explore whether some variant of this approach applies to the more general unitary evolutions we have considered here.

In this noninteracting setting, the unique decomposition of a unitary evolution into two components (as defined in the text) could be proved rigorously, allowing us to separate the dynamical topological behavior from the static topological behavior of the Floquet Hamiltonian. It is likely that this unitary decomposition is applicable more generally, including in interacting systems if many-body complications are dealt with appropriately. Indeed, we use this unitary decomposition as a working assumption in Ref.~\onlinecite{Roy:2016wd}, where it aids in the classification of Floquet SPTs in one dimension. This approach may be useful in the classification of driven, interacting topological phases more generally, a field in which much progress has recently been made \cite{Roy:2016wd,Else:2016tj,Potter:2016tb,vonKeyserlingk:2016vq}.

\begin{acknowledgements}
The authors would like to thank S.~L.~Sondhi for bringing this problem to our attention and sharing some particularly useful isights. The authors also thank D.~Reiss, C.~W.~von Keyserlingk and B.~Fregoso for fruitful discussions, and thank Michael Kolodrubetz for pointing out an error in an earlier version of this manuscript.
R.~R. and F.~H. acknowledge support from the NSF under CAREER DMR-1455368 and the Alfred P. Sloan foundation.
\end{acknowledgements}

\appendix
\begin{widetext}
\section{Action of Symmetry Operators on Unitaries\label{sec:symmetry_actions}}
In this appendix, we prove the action of the three symmetry operators on the time-evolution unitary. In order to simplify certain steps of the calculation, we will make use of the two-point (non-symmetrized) unitary operators defined through
\beq
U(\bk;t_2,t_1)&=&\mathcal{T}\exp\left(-i\int_{t_1}^{t_2}H(\bk,t')\dd t'\right),\label{eq:aux_unitaries}
\eeq
where we see that $U(\bk;t,0)\equiv U(\bk,t)$. These auxiliary unitaries satisfy the properties
\beq
\left[U(\bk;t_2,t_1)\right]^\dagger&=&U(\bk;t_1,t_2)\label{eq:aux_unitary_relations}\\
U(\bk;t_3,t_1)&=&U(\bk;t_3,t_2)U(\bk;t_2,t_1).\nonumber
\eeq

\subsection{Particle-hole Symmetry}
For the PHS operator, we start from
\beq
PH(\bk,t)P^{-1}&=&-H^*(-\bk,t)
\eeq
and find
\beq
PU(\bk,t)P^{-1}&=&P\left[\mathcal{T}\exp\left(-i\int_0^tH(\bk,t')\dd t'\right)\right]P^{-1}\nonumber\\
&=&\left[\sum_n\frac{(-i)^n}{n!}\mathcal{T}\int_0^t\dd t_1\ldots\int^t_0\dd t_n\,PH(\bk,t_1)P^{-1}\ldots PH(\bk,t_n)P^{-1}\right]\nonumber\\
&=&\left[\sum_n\frac{(+i)^n}{n!}\mathcal{T}\int_0^t\dd t_1\ldots\int^t_0\dd t_n\,H^*(-\bk,t_1)\ldots H^*(-\bk,t_n)\right]\nonumber\\
&=&\left[\mathcal{T}\exp\left(-i\int_0^tH(-\bk,t')\dd t'\right)\right]^*\nonumber\\
&=&U^*(-\bk,t).
\eeq
\subsection{Time-reversal Symmetry}
For the TRS operator, we start from
\beq
\theta H(\bk,t)\theta^{-1}&=&H^*(-\bk,T-t)
\eeq
and find
\beq
\theta U(\bk,t)\theta^{-1}&=&\theta\left[\mathcal{T}\exp\left(-i\int_0^tH(\bk,t')\dd t'\right)\right]\theta^{-1}\nonumber\\
&=&\left[\sum_n\frac{(-i)^n}{n!}\mathcal{T}\int_0^t\dd t_1\ldots\int^t_0\dd t_n\,\theta H(\bk,t_1)\theta ^{-1}\ldots \theta H(\bk,t_n)\theta^{-1}\right]\nonumber\\
&=&\left[\sum_n\frac{(-i)^n}{n!}\mathcal{T}\int_0^t\dd t_1\ldots\int^t_0\dd t_n\,H^*(-\bk,T-t_1)\ldots H^*(-\bk,T-t_n)\right]\nonumber\\
&\stackrel{=}{t_i\to T-t_i}&\left[\sum_n\frac{(+i)^n}{n!}\mathcal{T}\int_T^{T-t}\dd t_1\ldots\int^{T-t}_T\dd t_n\,H^*(-\bk,t_1)\ldots H^*(-\bk,t_n)\right]\nonumber\\
&=&\left[\mathcal{T}\exp\left(-i\int_T^{T-t}H(-\bk,t')\dd t'\right)\right]^*\nonumber\\
&=&U^*(-\bk;T-t, T).
\eeq
We rewrite this using Eq.~\eqref{eq:aux_unitary_relations} to obtain 
\beq
\theta U(\bk,t)\theta^{-1}&=&U^*(-\bk;T-t,0)U^*(-\bk;0,T)\\
&=&U^*(-\bk,T-t)U^{\dagger*}(-\bk,T).\nonumber
\eeq
\subsection{Chiral Symmetry}
For the CS operator, we start from
\beq
C H(\bk,t)C^{-1}&=&-H(\bk,T-t)
\eeq
and find
\beq
C U(\bk,t)C^{-1}&=&C\left[\mathcal{T}\exp\left(-i\int_0^tH(\bk,t')\dd t'\right)\right]C^{-1}\nonumber\\
&=&\left[\sum_n\frac{(-i)^n}{n!}\mathcal{T}\int_0^t\dd t_1\ldots\int^t_0\dd t_n\,C H(\bk,t_1)C ^{-1}\ldots C H(\bk,t_n)C^{-1}\right]\nonumber\\
&=&\left[\sum_n\frac{(+i)^n}{n!}\mathcal{T}\int_0^t\dd t_1\ldots\int^t_0\dd t_n\,H(\bk,T-t_1)\ldots H(\bk,T-t_n)\right]\nonumber\\
&\stackrel{=}{t_i\to T-t_i}&\left[\sum_n\frac{(-i)^n}{n!}\mathcal{T}\int_T^{T-t}\dd t_1\ldots\int^{T-t}_T\dd t_n\,H(\bk,t_1)\ldots H(\bk,t_n)\right]\nonumber\\
&=&\left[\mathcal{T}\exp\left(-i\int_T^{T-t}H(\bk,t')\dd t'\right)\right]\nonumber\\
&=&U(\bk;T-t,T).
\eeq
We rewrite this using Eq.~\eqref{eq:aux_unitary_relations} to obtain 
\beq
C U(\bk,t)C^{-1}&=&U(\bk;T-t,0)U(\bk;0,T)\\
&=&U(\bk,T-t)U^{\dagger}(\bk,T).\nonumber
\eeq
\section{Action of Symmetry Operators on Symmetrized Unitaries\label{sec:modified_symmetry_actions}}
In this appendix, we prove the action of the three symmetry operators on the symmetrized time-evolution unitaries $U_S(\bk,t)$ that are defined in Eq.~\eqref{eq:modifiedU_def}. We will derive these relations using the corresponding expressions for the original unitaries, which we derived previously in Appendix~\ref{sec:symmetry_actions}, and will also make use of the two-point unitaries defined in Eq.~\eqref{eq:aux_unitaries}. In particular, we note that
\beq
U_S(\bk,t)&=&U\left(\bk;\frac{1+t}{2},\frac{1-t}{2}\right)\nonumber\\
&=&U\left(\bk;\frac{1+t}{2},0\right)U\left(\bk;0,\frac{1-t}{2}\right)\nonumber\\
&=&U\left(\bk,\frac{1+t}{2}\right)\left[U\left(\bk,\frac{1-t}{2}\right)\right]^\dagger
\eeq
We will also make use of the symmetrized unitary relation $U_S^{\dagger}(\bk,t)=U_S(\bk,-t)$.
\subsection{Particle-hole Symmetry}
Starting from the unitary PHS relation
\beq
PU(\bk,t)P^{-1}&=&U^*(-\bk,t),
\eeq
we find that the symmetrized unitaries satisfy
\beq
PU_S(\bk,t)P^{-1}&=&PU\left(\bk;\frac{1+t}{2}\right)P^{-1}P\left[U\left(\bk;\frac{1-t}{2}\right)\right]^\dagger P^{-1}\nonumber\\
&=&U^*\left(-\bk,\frac{1+t}{2}\right)\left[P^{-1}U\left(\bk,\frac{1-t}{2}\right)P\right]^\dagger \nonumber\\
&=&U^*\left(-\bk,\frac{1+t}{2}\right)\left[U^*\left(-\bk,\frac{1-t}{2}\right)\right]^\dagger \nonumber\\
&=&U_S^*(-\bk,t).
\eeq
\subsection{Time-reversal Symmetry}
Starting from the unitary TRS relation
\beq
\theta U(\bk,t)\theta ^{-1}&=&U^*(-\bk,1-t)U^{\dagger *}(-\bk,1),
\eeq
we find that the symmetrized unitaries satisfy
\beq
\theta U_S(\bk,t)\theta^{-1}&=&\theta U\left(\bk,\frac{1+t}{2}\right)\theta^{-1}\theta\left[U\left(\bk,\frac{1-t}{2}\right)\right]^\dagger\theta^{-1}\nonumber\\
&=& U^*\left(-\bk,\frac{1-t}{2}\right)U^{\dagger*}\left(-\bk,1\right)\left[U^*\left(-\bk,\frac{1+t}{2}\right)U^{\dagger*}\left(-\bk,1\right)\right]^\dagger\nonumber\\
&=& U^*\left(-\bk,\frac{1-t}{2}\right)\left[U^*\left(-\bk,\frac{1+t}{2}\right)\right]^\dagger\nonumber\\
&=&U^{*}_S(-\bk,-t).
\eeq
Then, using the properties of symmetrized unitaries, this becomes
\beq
\theta U_S(\bk,t)\theta^{-1}=U_S^{*}(-\bk,-t)=U_S^{\dagger*}(-\bk,t).
\eeq
\subsection{Chiral Symmetry}
Starting from the unitary CS relation
\beq
C U(\bk,t)C ^{-1}&=&U(\bk,1-t)U^\dagger(\bk,1),
\eeq
we find
\beq
C U_S(\bk,t)C^{-1}&=&C U\left(\bk,\frac{1+t}{2}\right)C^{-1}C\left[U\left(\bk,\frac{1-t}{2}\right)\right]^\dagger C^{-1}\nonumber\\
&=& U\left(\bk,\frac{1-t}{2}\right)U^\dagger(\bk,1)\left[U\left(\bk,\frac{1+t}{2}\right)U^\dagger(\bk,1)\right]^\dagger\nonumber\\
&=& U\left(\bk,\frac{1-t}{2}\right)\left[U\left(\bk,\frac{1+t}{2}\right)\right]^\dagger\nonumber\\
&=&U_S(\bk,-t).
\eeq
Then, again using the properties of symmetrized unitaries, we obtain
\beq
C U_S(\bk,t)C^{-1}=U_S(\bk,-t)=U_S^{\dagger}(\bk,t).
\eeq
\end{widetext}

\section{Decomposition of Unitaries\label{sec:decomposition_proof}}
In this appendix, we prove the unitary decomposition theorem given in Sec.~\ref{sec:unitary_decomposition}, which is reproduced below.
\begin{theorem}
Every unitary $U\in\mathcal{U}^\mathcal{S}_{0,\pi}$ can be continuously deformed to a composition of a unitary loop $L$ and a constant Hamiltonian evolution $C$, which we write as $U\approx L*C$. $L$ and $C$ are unique up to homotopy.
\end{theorem}
The proof of this theorem has two stages. First, we show that there exists a decomposition $U \approx L*C$:
\begin{lemma}\label{lemma:loop_constant}
Every unitary $U\in\mathcal{U}^\mathcal{S}_{0,\pi}$ is homotopic to a product $L*C$, where $L$ is a unitary loop and $C$ is a constant evolution due to some static Hamiltonian (which is gapped at zero).
\end{lemma}
\begin{proof}
Let $H_F$ be the (unique) Floquet Hamiltonian for $U$ and let $C_\pm(s)$ be the constant evolution unitaries corresponding to the static Hamiltonians $\pm s H_F$. Consider the continuous family of unitaries 
\beq
h(s) &=& \left[U * C_-(s)\right]* C_+(s).
\eeq
It is clear that $h(0)$ is homotopic to $U$ and $h(1)$ is of the form $L * C_+(1)$ with $L =U * C_-™(1)$. The endpoint of $U * C_-™(1)$ is $U(1) \exp(iH_F) = \mathbb{I}$.
\end{proof}
Secondly, we show that the factors $L$ and $C$ involved in a decomposition $L*C$ are unique up to homotopy:
\begin{lemma}\label{lemma:loop_constant_homotopy}
Two compositions satisfy $L_1*C_1\approx L_2*C_2$ if and only if $L_1\approx L_2$ and $C_1\approx C_2$.
\end{lemma}
\begin{proof}
$L_1 * C_1\approx L_2 * C_2$ implies there is some function $h(s)$ for $s \in [0,1]$ such that $h(s)$ preserves the gap structure for all values of $s$ and
\beq
h(0) = L_1 * C_1,&~~~~~~&h(1) = L_2 * C_2.
\eeq
Let $H(s)$ be the Floquet Hamiltonian corresponding to the unitary $h(s)$. $H(s)$ then provides a homotopy between the Floquet Hamiltonians of $C_1$ and $C_2$.

Let $C_+(s)$ be the constant evolution unitary corresponding to the Hamiltonian $H(s)$. Since $H(s)$ is independent of time, $C_+(s)$ is a constant evolution unitary that continuously interpolates between $C_+(0) = C_1$ and $C_+(1) = C_2$. Thus, $C_1\approx C_2$.

Now, let $g(s) = h(s) * C_-(s)$, where $C_-™(s)$ is the constant Hamiltonian unitary with Hamiltonian $-H(s)$. $g(s)$ is a loop for all $s$ and interpolates between $L_1$ and $L_2$. Thus, $L_1\approx L_2$. The proof in the reverse direction follows trivially from the definition of homotopy.
\end{proof}

\section{Proof of One-to-one Mapping between Unitaries and Hermitian Maps\label{sec:121_proof}}
In this section, we prove the one-to-one correspondence between unitary evolutions and Hermitian maps defined according to Eq.~\ref{eq:hermitian_map_definition}. We give the proof for the case of PHS only, but note that the method may easily be extended to other symmetry classes.
\begin{claim}
There is a one-to-one mapping between the set of Hermitian matrix maps that satisfy
\beq
P_{1}H_U(\bk,t)P_{1}^{-1}&=&-H_U^*(-\bk,t)\label{eq:Hprop1}\\
P_{2}H_U(\bk,t)P_{2}^{-1}&=&H_U^*(-\bk,t)\label{eq:Hprop2}\\
H_U^2&=&\mathbb{I}\label{eq:Hprop3}
\eeq
and the set of unitary maps that satisfy $PU(\bk,t)P^{-1}=U^*(-\bk,t)$.
\end{claim}

\begin{proof}
For a given $U(\bk,t)$, such that $PU(\bk,t)P^{-1}=U^*(-\bk,t)$, let
\beq
H_U&=&\left(\begin{array}{cc}
0 & U(\bk,t)\\
U^\dagger(\bk,t) & 0
\end{array}\right)
\eeq
Then, with $P_{1}$ and $P_{2}$ as defined above, it is clear that Eqs.~\ref{eq:Hprop1}--\ref{eq:Hprop3} are satisfied. 

Conversely, for a given $H_U(\bk,t)$ that satisfies Eqs.~\ref{eq:Hprop1}--\ref{eq:Hprop2}, we note that
\beq
P_{1}P_{2}H_U(\bk,t)\left(P_{1}P_{2}\right)^{-1}&=&-H_U(\bk,t)\label{eq:Hprop4}
\eeq
with
\beq
P_{1}P_{2}&=&\left(\begin{array}{cc}
\mathbb{I} & 0 \\
0 & -\mathbb{I}
\end{array}\right)
\eeq
for Class D, where $P^2=\mathbb{I}$, and 
\beq
P_{1}P_{2}&=&\left(\begin{array}{cc}
-\mathbb{I} & 0 \\
0 & \mathbb{I}
\end{array}\right)
\eeq
for class D, where $P^2=-\mathbb{I}$. If
\beq
H_U&=&\left(\begin{array}{cc}
A & B \\
B^\dagger & D
\end{array}\right),
\eeq
then Eq.~\ref{eq:Hprop4} implies $A=D=0$, and Eq.~\ref{eq:Hprop3} implies $BB^\dagger=\mathbb{I}$, so that $B$ is unitary. We can then write 
\beq
H_U&=&\left(\begin{array}{cc}
0 & U(\bk,t) \\
U^\dagger(\bk,t) & 0
\end{array}\right),
\eeq
and from Eq.~\ref{eq:Hprop1} we see that $PU(\bk,t)P^{-1}=U^*(-\bk,t)$.
\end{proof}

\section{Additional K-Theory Details\label{app:KTheory}}
In this appendix, we give some additional details of the K-theory classification scheme outlined in the main text. For further information, we refer the reader to Refs.~\cite{Kitaev:2009vc,Karoubi:2023223,Atiyah:112984}.
\subsection{Grothendieck Group of Unitary Maps in a Symmetry Class\label{sec:group_unitaries}}
We consider the problem of classifying unitary maps on a manifold $M$ (for instance, $M=S^1\times S^1$ for a periodic unitary in 1D) in a general AZ symmetry class denoted by $\mathcal{S}$. We construct a group as follows: we take pairs $(U_1,U_2)$ and consider the operation `$+$' defined through
\beq
(U_1,U_2)+(U_3,U_4)&=& (U_1\oplus U_3,U_2\oplus U_4),\label{eq:define_addition}
\eeq
where $\oplus$ is the direct sum. We define the equivalence of pairs in the usual (stable homotopy) sense, and choose symmetry operators in such a way that a symmetry operator for the unitary $U_1\oplus U_3$ is the tensor sum of the corresponding symmetry operators for $U_1$ and $U_3$. The pairs then form an Abelian group under $+$, where the trivial element consists of the equivalence class of pairs of the form $(U,U)$. We denote this group by $K_U(\mcs,M)$.

\subsection{Categories and K-Theory for Classification of Unitaries}
In the main text we noted that the problem of classifying unitaries in symmetry class $\mcs$ is equivalent to the problem of classifying Hermitian maps (or `Hamiltonians') in some enhanced symmetry class $\mcs'$. Using the same reasoning as above, we can define an Abelian group of pairs of these Hamiltonians under the `$+$' operation, which we write as $K(\mcs',M)$.

Following Karoubi \cite{Karoubi:2023223}, for an arbitary Banach category, $\cat$, let us denote by $\mathscr{C}^{p,q}$ the category whose objects are the pairs $(E,\rho)$, where $E\in \mathrm{Ob}(\cat)$ and $\rho:C^{p,q}\to\mathrm{End}(E)$ is an $K$-algebra homomorphism, and where $K$ is $\mathbb{R}$ or $\mathbb{C}$ and $C^{p,q}$ is a real or complex Clifford algebra with $p$ negative generators and $q$ positive generators. A morphism from the pair $(E,\rho)$ to the pair $(E',\rho')$ is defined to be a $\cat$-morphism $f:E\to E'$ such that $f\cdot\rho(\lambda)=\rho(\lambda)\cdot f$ for each element $\lambda$ of $C^{p,q}$. 

To classify the Hamiltonians above, we now construct for every symmetry group $\mcs'$ two additive categories of the form $\mathscr{C}^{p,q}$ and $\mathscr{C}^{p',q'}$, where $p,q,p',q'$ all depend on $\mcs'$. Here, $\cat$ is a category which is either the category of \emph{Real} or complex vector bundles on $M$, or a closely related category (depending on $\mcs'$) \footnote{for definitions see \cite{Karoubi:2023223}}. If $\{S_i\}$ is the set of symmetry operators corresponding to $\mcs'$, then the canonical inclusion map from $\{S_i\}$ to $\{S_i,H\}$ leads to a quasi-surjective Banach functor $\phi':\cat^{p',q'}\to \cat^{p,q}$. This allows us to define a Grothendieck group $K(\phi')$ associated with this functor, such that the Grothendieck group $K(\mcs',M)$ is the same as $K(\phi')$.

The canonical inclusion map $\cat^{p,q}\subset\cat^{p,q+1}$ induces a quasi-surjective Banach functor $\phi:\cat^{p,q+1}\longrightarrow \cat^{p,q}$. When $\cat$ is the category of complex vector bundles on $M$, then the Grothendieck group $K(\phi)$ is denoted by $K^{p,q}(M)$, and when $\cat$ is the category of \emph{Real} vector bundles over the \emph{real} space $M$ \cite{Atiyah:1966hg}, then the Grothendieck group $K(\phi)$ is denoted by $KR^{p,q}(M)$. Here, the \emph{real} space $M$ corresponds to the existence of an involution which derives from $\bk\to-\bk$.

Using, repeatedly if necessary, the canonical Morita equivalences of the categories $\cat^{p,q}$ with $\cat^{p+1,q+1}$, and $\cat^{p,0}$ with $\cat^{0,p+2}$, we can establish equivalences between the categories $\cat^{p,q}$ and $\cat^{p',q'}$ for an arbitrary symmetry class $\mcs'$ and a category of the form $\tilde{\cat}^{p,q}$, where $\tilde{\cat}$ is the category of complex vector bundles over $M$ for $\mcs'\in\{\rm A,AIII\}$ and the category of \emph{Real} vector bundles over $M$ for all other symmetry classes. This then allows us to identify $K(\mcs',M)$ with some $KR^{0,q}(M)$ (with $0\leq q <8$) or some $K^{0,q}(M)$ (with $0\leq q <2$) and leads to the results in the main text. Further details will be presented elsewhere.

\end{document}